\documentclass[twoside,leqno,twocolumn]{article}

\usepackage[letterpaper]{geometry}

\usepackage{ltexpprt}
\usepackage{hyperref}

\usepackage{algorithm}
\usepackage[noend]{algorithmic}

\usepackage{multirow}
\usepackage{numprint}

\usepackage{amsmath}
\usepackage{amsfonts}

\DeclareRobustCommand{\frcshape}{\fontfamily{frc}\selectfont}
\DeclareTextFontCommand{\textfrc}{\frcshape}
\newcommand{\mathfrc}[1]{\mbox{\frcshape #1}}

\def\MdR{\ensuremath{\mathbb{R}}}

\newcommand{\set}[1]{\left\{ #1\right\}}
\newcommand{\sodass}{\,:\,}
\newcommand{\setGilt}[2]{\left\{ #1\sodass #2\right\}}

\newcommand{\While}    {{\bf while\ }}

\newcommand{\Do}       {{\bf do\ }}

\newcommand{\Is}       {:=}

\newcommand{\ie}{i.\,e.,\xspace}

\newcommand{\etal}{et~al.\xspace}

\newif\ifFull
\Fullfalse

\usepackage[square,numbers]{natbib}

\usepackage[usenames,dvipsnames]{xcolor}
\usepackage{hyperref}
\usepackage{xspace}
\usepackage{wrapfig}
\usepackage{graphics}
\usepackage{todonotes}
\usepackage{booktabs}

\usepackage{afterpage}

\begin{document}

\newcommand\relatedversion{}
\renewcommand\relatedversion{\thanks{The full version of the paper can be accessed at \protect\url{https://arxiv.org/abs/1902.09310}}} %

\title{\Large Faster Local Motif Clustering via Maximum Flows}
\author{Adil Chhabra\thanks{Heidelberg University, Germany.}
\and Marcelo Fonseca Faraj\thanks{Heidelberg University, Germany.}
\and Christian Schulz\thanks{Heidelberg University, Germany.}}

\date{}

\maketitle

\pagenumbering{arabic}

\begin{abstract} \small\baselineskip=9pt 	Local clustering aims to identify a cluster within a given graph that includes a designated seed node or a significant portion of a group of seed nodes. This cluster should be well-characterized, i.e.,~it has a high number of internal edges and a low number of external edges. 
In this work, we propose \texttt{SOCIAL}, a novel algorithm for local motif clustering which optimizes for motif conductance based on a local hypergraph model representation of the problem and an adapted version of the max-flow quotient-cut improvement algorithm~(\texttt{MQI}). 
In our experiments with the triangle motif, \texttt{SOCIAL} produces local clusters with an average motif conductance lower than the state-of-the-art, while being up to \hbox{multiple orders of magnitude faster}.
\end{abstract}

\section{Introduction}
\label{sec:introduction}

Graphs are fundamental for representing complex systems and relationships in a wide range of contexts. 
They can be used to model everything from data dependencies and social networks to web links and email interactions. 
With the massive expansion of data in recent years, many real-world graphs have grown to enormous sizes, making it challenging to analyze and understand them. 
In particular, many applications only require analyzing a small, localized portion of a graph rather than the entire graph, which is the case for community-detection on Web~\cite{epasto2014reduce} and social~\cite{jeub2015think} networks as well as structure-discovery in bioinformatics~\cite{voevodski2009spectral} networks, among others.
Those real-world applications are usually preceded by or modeled as a \emph{local clustering}. 
Local clustering aims at identifying a specific cluster within a given graph that includes a designated seed node or a portion of a group of seed nodes, and is \emph{well-characterized}, i.e., it consists of many internal edges and few external edges.
More specifically, the quality of a community can be quantified by specific metrics such as \emph{conductance}~\cite{kannan2004clusterings}.
Since minimizing conductance is NP-hard~\cite{wagner1993between}, approximative and heuristic approaches are used in practice.
Given the nature and scale of this problem, these approaches should ideally require time and memory usage dependent only \hbox{on the size of the found cluster}.

The local clustering problem has been investigated both theoretically~\cite{andersen2006local} and experimentally~\cite{leskovec2009community}, and has been solved using a wide variety of techniques, including statistical~\cite{chung2013solving,kloster2014heat}, numerical~\cite{li2015uncovering,mahoney2012local}, and combinatorial~\cite{orecchia2014flow,fountoulakis2020flowbased} methods.
While traditional approaches to local clustering typically consider the edge distribution when evaluating the quality of a local community, novel methods~\cite{yin2017local,zhang2019local,meng2019local,murali2020online,LocMotifClusHyperGraphPartition} have shifted focus to finding local communities based on the distribution of \emph{motifs}, higher-order structures within the graph.
These works provide empirical evidence that this approach, which can be called \emph{local motif clustering}, is effective at detecting high-quality local communities.
Nevertheless, since this local clustering perspective is relatively new, there are still many opportunities to improve upon current approaches and discover more efficient algorithms for finding high-quality solutions.

\vspace*{-0.2cm}
\subsubsection*{Contribution.} In this work, we propose a novel algorithm for local motif clustering which optimizes for motif conductance by combining the strongly local hypergraph model from~\hbox{\citet{LocMotifClusHyperGraphPartition}} with an adapted version of the fast and effective algorithm \emph{max-flow quotient-cut improvement}~(\texttt{MQI})~\cite{mqipaper2004}.
Our algorithm \texttt{SOCIAL}, which stands for \emph{faSter mOtif Clustering vIa mAximum fLows}, starts by building a hypergraph model which is an exact representation for the motif-distribution around the seed node on the original graph~\cite{LocMotifClusHyperGraphPartition}.
Using this model, we create a flow model in which certain cuts correspond one-to-one with sub-sets of the initial cluster that include the seed node and have lower motif conductance than that of the whole cluster. 
We then use a push-relabel algorithm to either find such a cut and repeat the process recursively, or to prove that the current cluster is optimal among all its sub-clusters containing the seed node. 
In our experiments with the triangle motif, \texttt{SOCIAL} produces communities with a motif conductance value that is lower than the state-of-the-art, while also being up to multiple orders of magnitude faster.

\section{Preliminaries}
\label{sec:preliminaries}

\label{subsec:basic_concepts}

\subsubsection*{Graphs.}
Let $G=(V=\{0,\ldots, n-1\},E)$ be an \emph{undirected graph} with no multiple or self edges allowed, such that $n = |V|$ and $m = |E|$.
Let $c: V \to \MdR_{\geq 0}$ be a node-weight function, and let $\omega: E \to \MdR_{>0}$ be an edge-weight function.
We generalize $c$ and $\omega$ functions to sets, such that $c(V') = \sum_{v\in V'}c(v)$ and $\omega(E') = \sum_{e\in E'}\omega(e)$.
Let $N(v) = \setGilt{u}{\set{v,u}\in E}$ be the \emph{open neighborhood} of $v$, and let $N[v]=N(v) \cup \{v\}$ be the \emph{closed neighborhood} of $v$.
We generalize the notations $N(.)$ and $N[.]$ to sets, such that $N(V') = \cup_{v\in V'}N(v)$ and $N[V'] = \cup_{v\in V'}N[v]$.
A graph $G'=(V', E')$ is said to be a \emph{subgraph} of $G=(V, E)$ if $V' \subseteq V$ and $E' \subseteq E \cap (V' \times V')$. 
When $E' = E \cap (V' \times V')$, $G'$ is the subgraph \emph{induced} in $G$ by $V'$.
Let $\overline{V'} = V \setminus V'$ be the \emph{complement} of a set $V' \subseteq V$ of nodes. 
Let a \emph{motif} $\mu$ be a connected graph.
\emph{Enumerating} the motifs $\mu$ in a graph $G$ consists building the collection $M$ of all occurrences of $\mu$ as a subgraph of $G$.
Let $d(v)$ be the \emph{degree} of node $v$ and $\Delta$ be the maximum degree of $G$.
Let $d_\omega(v)$ be the \emph{weighted degree} of a node $v$ and $\Delta_\omega$ be the maximum weighted degree of $G$.
Let $d_\mu(v)$ be the \emph{motif~degree} of a node~$v$, i.e., the number of motifs $\mu \in M$ which contain $v$.
We generalize the notations $d(.)$, $d_\omega(.)$, and $d_\mu(.)$ to sets, such that the \emph{volume} of~$V'$ is~$d(V') = \sum_{v\in V'}d(v)$, the \emph{weighted~volume} of~$V'$ is~$d_\omega(V') = \sum_{v\in V'}d_\omega(v)$, and the \emph{motif~volume} of~$V'$ is~$d_\mu(V') = \sum_{v\in V'}d_\mu(v)$.
Let a \emph{spanning forest} of $G$ be an acyclic subgraph of $G$ containing all its nodes.
Let the \emph{arboricity} of $G$ be the minimum amount of spanning forests of $G$ \hbox{necessary to cover all its edges}.

\vspace*{-0.2cm}
\subsubsection*{Local Motif Clustering.}
In the \emph{local graph clustering} problem, a graph $G=(V,E)$ and a seed node $u \in V$ are taken as input and the goal is to detect a \emph{well-characterized cluster} (or \emph{community}) $C \subset V$ containing~$u$.
A high-quality cluster $C$ usually contains nodes that are densely connected to one another and sparsely connected to $\overline{C}$.
There are many functions to quantify the quality of a cluster, such as \emph{modularity}~\cite{brandes2007modularity} and \emph{conductance}~\cite{kannan2004clusterings}.
The conductance metric is defined as $\phi(C)=|E'|/\min(d(C),d(\overline{C}))$, where $E'=E \cap (C \times \overline{C})$ is the set of edges shared by a cluster $C$ and its complement.
\emph{Local motif graph clustering} is a generalization of local graph clustering where a motif $\mu$ is taken as an additional input and the computed cluster optimizes a clustering metric based on $\mu$.
In particular, the \emph{motif conductance} $\phi_\mu(C)$ of a cluster $C$ is defined by 
\hbox{\citet{benson2016higher}} as a generalization of the conductance in the following way:
$\phi_\mu(C)=|M'|/min(d_\mu(C),d_\mu(\overline{C}))$, where $M'$ are all the motifs $\mu$ which contain at least one node in $C$ and one node in $\overline{C}$.
Note that, if the motif under consideration is simply an edge, then $|M'|$ is the \hbox{edge-cut and $\phi_\mu(C)=\phi(C)$}.

\vspace*{-0.2cm}
\subsubsection*{Hypergraphs.}
Let $H=(\mathcal{V}=\{0,\ldots, \mathfrc{n}-1\},\mathcal{E})$ be an \emph{undirected hypergraph} with no multiple or self hyperedges allowed, with $\mathfrc{n} = |\mathcal{V}|$ \emph{nodes} and $\mathfrc{m} = |\mathcal{E}|$ \emph{hyperedges} (or \emph{nets}).
A net is defined as a subset of $\mathcal{V}$.
The nodes that compose a net are called \emph{pins}.
Let $\mathfrc{c}: \mathcal{V} \to \MdR_{\geq 0}$ be a node-weight function, and let $\mathfrc{w}: \mathcal{E} \to \MdR_{>0}$ be a net-weight function.
We generalize $\mathfrc{c}$ and $\mathfrc{w}$ functions to sets, such that $\mathfrc{c}(\mathcal{V}') = \sum_{v\in \mathcal{V}'}\mathfrc{c}(v)$ and $\mathfrc{w}(\mathcal{E}') = \sum_{e\in \mathcal{E}'}\mathfrc{w}(e)$.
A node $v\in\mathcal{V}$ is \emph{incident} to a net $e\in\mathcal{E}$ if $v \in e$.
Let~$\mathcal{I}(v)$ be the set of incident nets of~$v$, let~$\mathfrc{d}(v) \Is |\mathcal{I}(v)|$ be the \emph{degree} of~$v$, and let $\mathfrc{d}_{\mathfrc{w}}(v) \Is \mathfrc{w}(\mathcal{I}(v))$ be the \emph{weighted degree} of $v$.
We generalize the notations $\mathfrc{d}(.)$ and  $\mathfrc{d}_{\mathfrc{w}}(.)$ to sets, such that the \emph{volume} of~$\mathcal{V}'$~is~$\mathfrc{d}(\mathcal{V}') = \sum_{v\in \mathcal{V}'}\mathfrc{d}(v)$ and the \emph{weighted~volume} of~$\mathcal{V}'$ is~$\mathfrc{d}_{\mathfrc{w}}(\mathcal{V}') = \sum_{v\in \mathcal{V}'}\mathfrc{d}_{\mathfrc{w}}(v)$.
Two nodes are \emph{adjacent} if they are incident to a same net.
Let the number of pins~$|e|$ in a net~$e$ be the \emph{size} of~$e$.
We define the \emph{contraction} operator as~$\big/$ such that $H \big/ \mathcal{V}^\prime$, with $\mathcal{V}^\prime \subseteq \mathcal{V}$, is the hypergraph obtained by contracting the nodes from~$\mathcal{V}^\prime$ on~$H$.
This contraction consists of substituting all the nodes in $\mathcal{V}^\prime$ by a single representative node $x$, removing nets totally contained in $\mathcal{V}^\prime$, and substituting all the pins in~$\mathcal{V}^\prime$ by a single pin~$x$ in each of the remaining nets.
Given a cluster~$\mathcal{V}^\prime~\subseteq~\mathcal{V}$, the \emph{cut} or \emph{cut-net}~$cut(\mathcal{V}^\prime)$ of~$\mathcal{V}^\prime$ consists of the total weight of the nets crossing the cluster, i.e., $cut(\mathcal{V}^\prime) = \sum_{e \in \mathcal{E}^\prime}\mathfrc{w}(\mathcal{E}^\prime)$, in which $\mathcal{E}^\prime \Is $ $\big\{e \in \mathcal{E} : \exists i,j \mid e \cap \mathcal{V^\prime} \neq \emptyset, e \cap \mathcal{\overline{V~\prime}} \neq \emptyset , i\neq j\big\}$.

\vspace*{-0.2cm}
\subsubsection*{Flows.}
Let~$G_f=(V_f,E_f)$ be a flow graph.
A flow graph has one source node~$s \in V_f$, one sink node~$t \in V_f$, and a set of remaining nodes $V \setminus \{s,t\}$.
All edges~$e=(u,v)$ in a flow graph are directed and associated with a nonnegative capacity~$cap(u,v)$.
An \hbox{s-t flow} is a function \hbox{$f:V_f\times{V_f}\rightarrow\MdR_{>0}$} which satisfies a \emph{capacity} constraint, \ie $f(u,v) \leq cap(u,v)$, a \emph{symmetry} constraint, \ie $\forall{u,v}\in{V_f}: f(u,v)=-f(v,u)$, and a \emph{flow conservation} constraint, \ie \hbox{$\forall{u}\in{V_f}\setminus\{s,t\}:$} \hbox{$\sum_{v\in V_f}f(u,v)=0$}.
An edge~$(u,v)$ is called~\emph{saturated} if \hbox{$cap(u,v)=f(u,v)$};
The total amount of flow moved from~$s$~to~$t$ is defined as the \emph{value}~$|f|$ of~$f$ and is computed as follows: \hbox{$|f|=\sum_{u\in V_f}f(u,t)=\sum_{v\in V_f}f(s,v)$}.
A given \hbox{s-t~flow~$f$} in~$G_f$ is \emph{maximum} if, for any \hbox{s-t~flow}~$f'$ in~$G_f$,~$|f'|\leq|f|$.
Let~$G_r=(V_f,E_r)$ be the \emph{residual graph} associated with a given flow~$f$~on~$G_f$, such that $E_r=$ \hbox{$\{(u,v) \in V_f \times V_f: cap(u,v)-f(u,v)>0\}$}.
According to the Max-Flow Min-Cut Theorem~\cite{ford_fulkerson_1956}, the value~$|f|$ of a maximum s-t flow~$f$ on~$G_f$ equals the weight of a minimum s-t cut on~$G_f$, \ie a 2-way partition of~$G_f$ where edge weights equal edge capacities, $s$~and~$t$ are in distinct blocks, and the total weight of the cut edges is minimum.
To find the sink side of the minimum cut associated with a maximum flow in $G_f$, a reverse breadth-first search can be performed~on~$G_f$ \hbox{starting at the~sink~node~$t$}.

\vspace*{-0.3cm}
\subsubsection*{Push-Relabel.} 
For each node~$u$ in a flow graph~${V_f}$, let~\hbox{$exc(u)=\sum_{v\in V_f}f(u,v)$} be its \emph{excess} value and~$d(u)$ be its potential.
A node~$u$ is called \emph{active} if \hbox{$exc(u)>0$}.
An edge~$(u,v)$ is called~\emph{admissible} if \hbox{$cap(u,v)-f(u,v)>0$} and~\hbox{$d(u)=d(v)+1$}.
The push-relabel~\cite{push_relabel88} algorithm builds a maximum flow by computing a succession of \emph{preflows}, \ie~flows where the flow conservation constraint is substituted by \hbox{$\forall{u}\in{V_f}\setminus\{s,t\}:$} \hbox{$exc(u)\geq0$}. 
In the initial preflow, all out-edges of~$s$ are saturated, $\forall{u}\in{V_f}\setminus\{s\}:d(u)=0$, and $d(s)=|V_f|$.
The initial preflow is evolved via operations \emph{push}, \ie sending as much flow as possible from an active node through an admissible edge, and \emph{relabel}, \ie increasing the potential of a node until it~becomes~active.
Preflows induce minimum sink-side cuts, so a maximum flow and a minimum cut are obtained once \hbox{no node is active.}

\subsubsection*{Flows on Hypergraphs.}
A common technique to solve flow and cut problems on hypergraphs consists of transforming them in directed graphs and then applying traditional graph-based techniques on~them.
Among the existing transformations~\cite{veldt2022hypergraph,lawler1973cutsets}, we highlight \emph{clique expansion}, \emph{star expansion}, and \emph{Lawler expansion}.
In the \emph{clique expansion}, each net is represented by a clique, i.e., a set of edges connecting each pair of its pins in both directions.
In this approach, the weight of each edge is equal to weight of the corresponding net~$e$ divided by~$|e|-1$ and parallel edges are substituted by a single edge whose weight is the sum of the weights of the removed edges.
In the \emph{star expansion}, each net is represented by an auxiliary artificial node connected to its pins by edges in both directions. 
In this expansion, the edges have the same weight as the corresponding net.
In the \emph{Lawler expansion}, each net~$e$is represented by two auxiliary artificial nodes~$w_1$~and~$w_1$ and a collection of edges.
In particular, there is a directed edge $(w_1,w_2)$ which has the same weight as the corresponding net.
Additionally, each pin of the corresponding net has an out-edge to~$w_1$ and an in-edge from~$w_2$, each of them with weight infinity.
The three transformation approaches are exemplified in Figure~\ref{fig:hyperedge_expansion}.

\begin{figure}[t]
	\centering
	\includegraphics[width=1\linewidth]{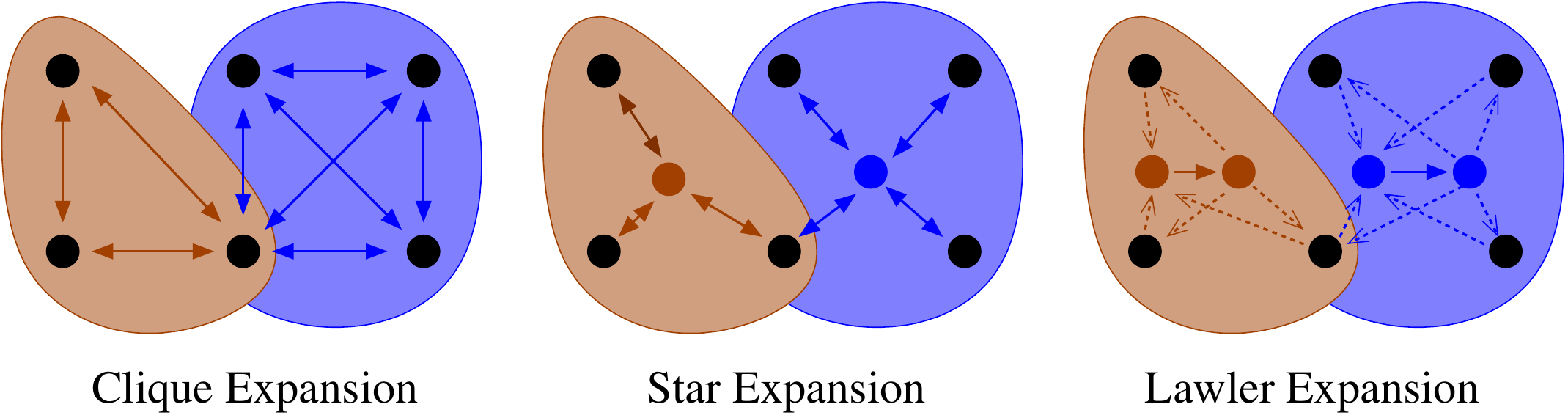}
	\vspace*{-0.5cm}
	\caption{Three existing approaches to represent a hypergraph using a directed graph. Nodes and nets of the original hypergraph are respectively represented by black circles and colored areas around them. Auxiliary artificial nodes and edges are respectively represented by circles and arrows, both with the same color as the corresponding net. Bidirectional arrows represent a pair of edges in both directions. Solid edges have finite weight while dashed edges have \hbox{infinite weight}.}
	\label{fig:hyperedge_expansion}
	\vspace*{-0.3cm}
\end{figure}

\subsection{Related Work.}
\label{subsec:related_work}

Motif-based clustering has been widely studied in the literature, with works such as~\cite{benson2015tensor, yin2017local,klymko2014using,prvzulj2007biological,tsourakakis2017scalable} partitioning all the nodes of a graph into clusters based on motifs.
We also address the topic of clustering based on motifs, but our focus is on identifying clusters in the immediate vicinity of a specific seed node, rather than on the entire graph.
Several works~\cite{kloster2014heat,li2015uncovering,mahoney2012local,cui2014local,sozio2010community} propose local clustering algorithms, but they do not focus on optimizing for motif-based metrics like our work. 
Instead, they use metrics based on edges, like conductance and modularity.
In this section, we review previous work on local clustering based on motifs, which is the focus~of~our~work.

\citet{rohe2013blessing} propose a local clustering algorithm based on triangle motifs.
Their algorithm starts by initializing a cluster containing only the seed node, and iteratively grows this cluster.
Particularly, the algorithm greedily inserts nodes contained in at least a predefined amount of cut triangles.
\hbox{\citet{huang2014querying}} recover local communities containing a seed node in online and dynamic setups based on higher-order graph structures named Trusses~\cite{cohen2008trusses}.
They define the $k$-truss of a graph as its largest subgraph whose edges are all contained in at least $(k-2)$ triangle motifs, hence trusses are a graph structure based on the frequency of triangles.
The authors use indexes to search for $k$-truss communities in time proportional to the size of the \hbox{recovered~community}. 

\citet{yin2017local}~propose \texttt{MAPPR}, a local motif clustering algorithm based on the Approximate Personalized PageRank~(\texttt{APPR}) method.
In a preprocessing phase, \texttt{MAPPR} enumerates the motif of interest in the entire input graph and constructs a weighted graph~$W$, in which edges only exist between nodes that appear in at least one instance of the motif, and their edge weight is equal to the number of occurrences of the motif containing these two endpoints.
Afterward, \texttt{MAPPR} uses an adapted version of the \texttt{APPR} method to find local communities in the weighted graph constructed in the preprocessing phase.
\texttt{MAPPR} is able to extract local communities from directed input graphs, something that cannot be \hbox{done using~\texttt{APPR}~alone}.

\citet{zhang2019local} propose \texttt{LCD-Motif}, an algorithm that addresses the local motif clustering problem using a modified version of the spectral method.
\texttt{LCD-Motif} has two main differences in comparison to the traditional spectral motif clustering method.
First, instead of computing singular vectors, the algorithm performs random walks to identify potential members of the searched cluster.
They use the span of a few dimensions of vectors, obtained through random walks, as an approximation for the local motif spectra.
Second, Instead of using $k$-means for clustering, \texttt{LCD-Motif} searches for the minimum 0-norm vector within the previously mentioned span, which must contain the seed nodes in its~support~vector.

\citet{meng2019local} propose \texttt{FuzLhocd}, a local motif clustering algorithm that uses fuzzy arithmetic to optimize a modified version of modularity.
Given seed node, \texttt{FuzLhocd} starts by detecting probable core nodes of the targeted local community using fuzzy membership.
After identifying the probable core nodes of the target local community using fuzzy membership, the algorithm expands these nodes using another fuzzy membership to form~a~cluster.

\citet{zhou2021high} propose \texttt{HOSPLOC}, a local motif clustering algorithm that uses a motif-based random walk to compute a distribution vector, which is then truncated and used in a vector-based partitioning method.
The algorithm begins by approximately estimating the distribution vector through a motif-based random walk. 
To further refine the computation and focus on the local region, \texttt{HOSPLOC} sets all small vector entries to 0. 
After this preprocessing step, the algorithm applies a vector-based partitioning method~\cite{spielman2013local} on the resulting distribution vector in order to \hbox{identify~a~local~cluster}.

\citet{shang2022local} propose \texttt{HSEI}, a local motif clustering algorithm that uses motif and edge information to grow a cluster from a seed node.
The algorithm begins by creating an initial cluster consisting of only the seed node.
It then adds nodes to the cluster from the seed's neighborhood, selecting them based on their motif degree.
The cluster is expanded using a motif-based extension of the~modularity~function.

\citet{LocMotifClusHyperGraphPartition} propose an algorithm to solve the local motif clustering problem using powerful (hyper)graph partitioning tools~\cite{kaffpa,schlag2016k,gottesburen2021scalable,DBLP:conf/esa/GottesburenH00S21}.
Their algorithm first uses a breadth-first search to select a ball containing the seed node and nearby nodes.
Next, they enumerate motif occurrences within the ball and build a (hyper)graph model which allows them to compute the motif conductance of any cluster within the ball.
They then partition their model into two blocks using a high-quality (hyper)graph partitioning algorithm, and refine the solution for motif conductance.

\section{Local Motif Clustering via Maximum Flows}
\label{sec:Faster Local Motif Clustering via Maximum Flows}

We now present the overall clustering strategy of \texttt{SOCIAL}, then we discuss its \hbox{algorithmic components}.

\subsection{Overall Strategy.}
\label{subsec:Overall Strategy}

Given a graph $G=(V,E)$, a seed node $u$, and a motif $\mu$, our strategy for local clustering is based on the following phases.
First, we select a set $S \subseteq V$ containing~$u$ and close-by nodes.
From now on, we refer to this set $S$ as a \emph{ball around}~$u$.
Second, we enumerate the collection $M$ of occurrences of the motif $\mu$ which contain at least one node in~$S$.
Third, we build a hypergraph model $H_\mu$ in such a way that the motif-conductance of any cluster~$C \subseteq S$ in~$G$ can be computed directly in~$H_\mu$.
Fourth, we set~$C_0=S$ as our initial cluster and use it to build our \texttt{MQI}-based~\cite{mqipaper2004} flow model~$G_f$ from the hypergraph model~$H_\mu$.
Fifth, we use~$G_f$ to either find a new cluster~$C \subset C_0$ containing~$u$ with strictly smaller motif conductance than~$C_0$ or prove that such cluster does not exist.
While~$C \subset C_0$ is found, we take it as our new initial cluster, rebuild~$G_f$, and repeat the previous~phase.
When eventually no such strict sub-set is found, the best obtained cluster is directly translated back to $G$ as a local cluster around the seed node.
Figure~\ref{fig:overall_algorithm} provides a comprehensive illustration of the consecutive phases of \texttt{SOCIAL}.
Note that there is no guarantee of finding the best overall cluster including~$u$ strictly contained in~$S$.
Instead, we find a succession of clusters with strictly decreasing cardinality and motif conductance until a local optimum~is~reached.
To better explore the vicinity of~$u$~in~$G$ and overcome the fact we only find clusters inside~$S$, we repeat the overall strategy~$\alpha$~times with distinct balls~$S$.
Our overall algorithm including the mentioned repetitions is outlined in Algorithm~\ref{alg:overall_strategy}.

\begin{figure*}[t]
	\centering
	\includegraphics[width=0.9\linewidth]{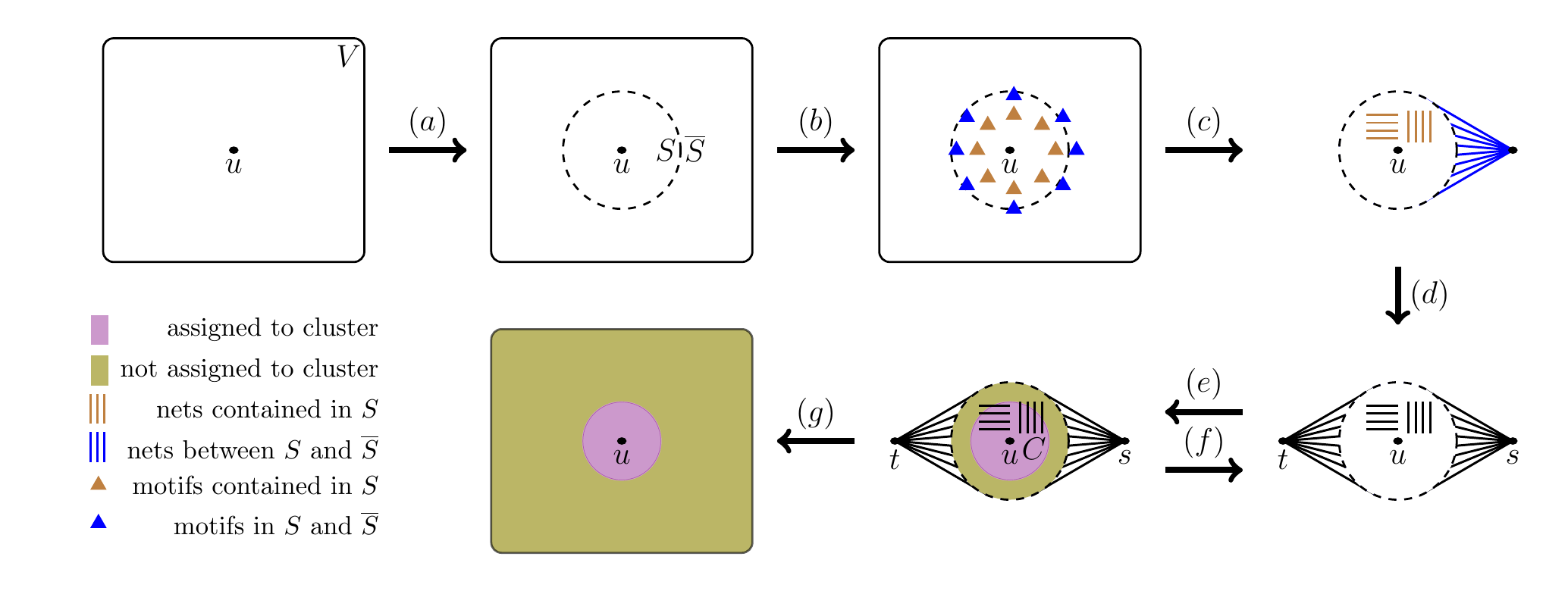}
	\vspace*{-0.3cm}
	\caption{Illustration of the phases of \texttt{SOCIAL}. (a)~Given a seed node $u$ and a graph $G$, a ball $S$ around $u$ is selected. (b)~Motif occurrences of $\mu$ with at least a node in $S$ are enumerated. (c)~The hypergraph model $H_\mu$ is built by converting motifs into nets and contracting $\overline{S}$ into a single node. The ball~$S$ is taken as the initial cluster~$C_0$. (d)~The flow model~$G_f$ is built based on~$C_0$~in~$H_\mu$. (e)~A cluster~$C \subseteq C_0$ containing~$u$ is found using maximum flows. (f)~While~$C \subset C_0$, the model~$G_f$ is rebuilt based on~$C$, which is taken as the initial cluster~$C_0$. (g)~When eventually~$C = S$,~$C$ is converted in a local cluster around the seed~node~in~$G$.}
	\label{fig:overall_algorithm}
\end{figure*}

\begin{algorithm}[t]
	\caption{Local Motif Clustering via Max Flows}
	\label{alg:overall_strategy}
	\hspace*{-0cm} \textbf{Input} graph $G=(V,E)$; seed node $u \in V$; motif $\mu$ \\
	\hspace*{-0cm} \textbf{Output} cluster $C^* \subseteq V$ 
	\begin{algorithmic}[1]  
		\STATE $C^* \leftarrow \emptyset$
		\FOR{$i=1,\ldots,\alpha$}
		\STATE Select ball $S$ around $u$
		\STATE $M \leftarrow$ Enumerate motifs in $S$			
		\STATE Build hypergraph model $H_\mu$ based on~$S$~and~$M$%
		\STATE $C \leftarrow S$
		\STATE\Do
		\STATE \hskip1.0em $C_0 \leftarrow C$
		\STATE \hskip1.0em Build flow model~$G_f$ based on~$C_0$~in~$H_\mu$
		\STATE \hskip1.0em Solve~$G_f$ to obtain cluster~$C \subseteq C_0$ including~$u$ \\
		\While{$C \subset C_0$}
		\IF{$C^* = \emptyset \lor \phi_\mu(C) <  \phi_\mu(C^*)$} 
		\STATE $C^* \leftarrow C$
		\ENDIF
		\ENDFOR
		\STATE Convert $C^*$ into a local motif cluster in $G$
	\end{algorithmic}

\end{algorithm}

\subsection{Hypergraph Model.}
\label{subsec:Hypergraph Model}

We follow the same procedure as~\hbox{\citet{LocMotifClusHyperGraphPartition}} to construct the hypergraph model~$H_\mu$.
To ensure a thorough understanding of our overall algorithm, we provide a summary of the phases involved, i.e., finding a ball around the seed node, enumerating motifs within it, and finally constructing the hypergraph model~$H_\mu$.

\subsubsection*{Ball around the Seed Node.}
Our approach to select a ball~$S$ is a fixed-depth breadth-first search (BFS) rooted on $u$.
More specifically, we compute the first $\ell$ layers of the BFS tree rooted on $u$, then we include all its nodes in~$S$.
For each of the $\alpha$ repetitions of the overall algorithm, we use different amounts $\ell$ of layers for a better algorithm exploration. 
Two special cases are handled by \texttt{SOCIAL}, namely a ball~$S$ that is either too small or disconnected from $\overline{S}$.
We avoid the first special case by ensuring that~$S$ contains $100$ or more nodes in at least one repetition of our overall algorithm.
More specifically, in case this condition is not automatically met, then we accomplish it in the last repetition by growing additional layers in our partial BFS tree while it contains fewer than $100$ nodes.
The number $100$ is based on the findings of \hbox{\citet{leskovec2009community}}, which show that most well characterized communities from real-world graphs have a relatively small size, in the order of magnitude of $100$ nodes.
If the second exceptional case happens, it means that the whole BFS tree rooted on the seed node has at most~$\ell$ layers.
In this case, we simply stop the algorithm and return the entire ball $S$, which corresponds to an optimal community with motif conductance $0$ provided that there is at least one motif in $S$.
The number~$\alpha$ of repetitions as well as the amount $\ell$ of layers used in each repetition are tuning~parameters.

\vspace*{-0.2cm}
\subsubsection*{Motif Enumeration.}
Although enumerating a general motif on some graph is NP-hard~\cite{read1977graph}, there are efficient heuristics to do it such as the one proposed by \hbox{\citet{kimmig2017shared}}.
Nevertheless, simpler motifs such as small paths, cycles, and cliques can be trivially enumerated in polynomial time.
We focus our enumeration phase on the triangle motif. 
We implement the simple and exact algorithm proposed by \hbox{\citet{chiba1985arboricity}} to enumerate the collection $M$ of occurrences of the motif $\mu$ which contain at least one node in~$S$.
Roughly speaking, this algorithm works by intersecting the neighborhoods of adjacent nodes.
For each node $v$, the algorithm starts by marking its neighbors with degree smaller than or equal to its own degree.
For each of these specific neighbors of $v$, it then scans its neighborhood and enumerates new triangles as soon as marked nodes are found.
The running time of this algorithm is $O(ma) = O(m^{\frac{3}{2}})$, where $a$ is the arboricity of the graph.
We apply this enumeration algorithm only on the subgraph induced in $G$ by $N[S]$, which is enough to find all triangles containing at least one node in $S$, as exemplified by transformation~(a) in Figure~\ref{fig:model_construction}. 
Assuming a constant-bounded arboricity, the overall cost of our motif-enumeration phase for triangles~is~$O\big(|N[S] \times N[S])\cap E|\big)$.

\vspace*{-0.2cm}
\subsubsection*{Hypergraph Model.} 
The hypergraph model~$H_\mu$ is finally built in two conceptual operations.
First, define a hypergraph containing $V$ as nodes and a set~$\mathcal{E}$ of nets such that, for each motif in~$M$, $\mathcal{E}$ has a net with pins equal to the endpoints of this motif. 
Then, we contract together all nodes in~$\overline{S}$ into a single node~$r$ and substitute parallel nets by a single net whose weight is equal to the summed weights of the removed parallel nets.
More formally, we define the hypergraph version of our model as $H_\mu = (S \cup \{r\},\mathcal{E})$ where the set~$\mathcal{E}$ of nets contains one net~$e$ associated with each motif occurrence $G'=(V',E') \in M$ such that $e = V'$ if $V'\subseteq S$, and $e = V' \cap S \cup \{r\}$ otherwise.
In the former case the net has weight~$1$, in the latter case the net has weight equal to the amount of motif occurrences in $M$ represented by it.
Since node weights in~$H_\mu$ are irrelevant for \texttt{SOCIAL}, the involved theorems, and the motif conductance metric, we make all node weights unitary in~$H_\mu$.
In practice, the model~$H_\mu$ can be built by instantiating the nodes in $S \cup \{r\}$ and the nets in $\mathcal{E}$.
Assuming that the number of nodes in $\mu$ is a constant, our model is built in time $O(|S|+|M|)$ and uses memory $O(|S|+|M|)$.
The construction of~$H_\mu$ is illustrated in transformation~(c) of Figure~\ref{fig:overall_algorithm} and demonstrated for a particular example in transformation~(b)~of~Figure~\ref{fig:model_construction}.
Theorem~\ref{theo:conductance_equivalence} shows that the motif conductance~in~$G$ of any cluster~$C \subseteq S$ can be directly computed from~$H_\mu$ assuming~$d_\mu(S) \leq d_\mu(\overline{S})$.
The assumption $d_\mu(S) \leq d_\mu(\overline{S})$ is fair in practice since the ball~$S$ computed via BFS tends to be considerably smaller than~$\overline{S}$ for huge sparse networks.
Enumerating the motifs in $\overline{S}$ is not reasonable for a local clustering algorithm, but we did verify that our assumption holds during all~our~experiments.

\begin{figure}[t]
	\centering
	\includegraphics[width=\linewidth]{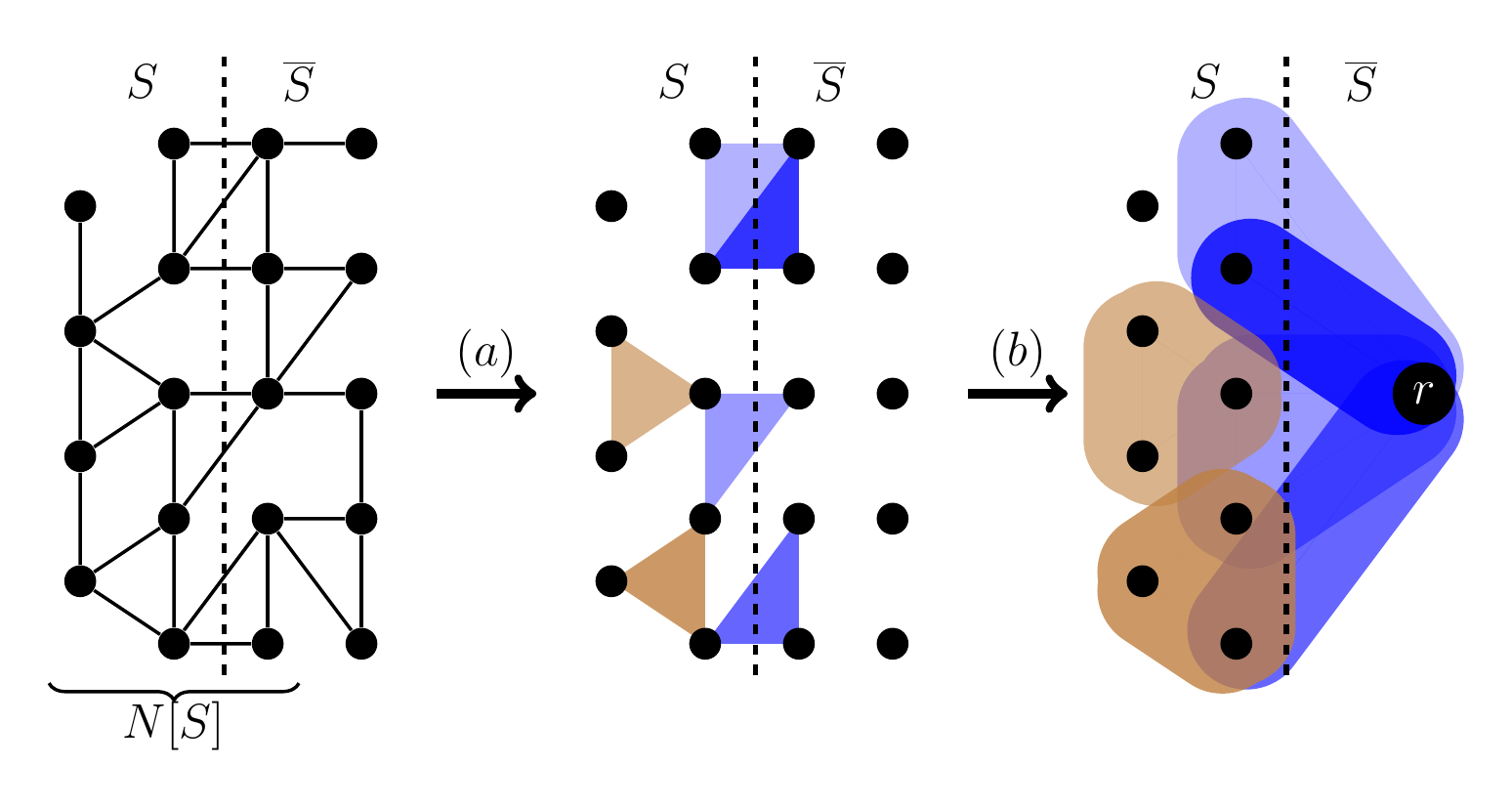}
	\vspace*{-0.6cm}
	\caption{Example of motif-enumeration and model-construction phases of \texttt{SOCIAL} for the triangle motif. In the left, the nodes of~$G$ are split into sets~$S$~and~$\overline{S}$. In the center, motif occurrences containing nodes in~$S$ are enumerated. In the right,~$H_\mu$ is built by converting motifs in nets and contracting~$\overline{S}$ into a node~$r$.}
	\label{fig:model_construction}
\end{figure}

\vspace*{-0.1cm}
\begin{theorem}[Theorem 3.2 from~\cite{LocMotifClusHyperGraphPartition}]
	The motif conductance $\phi_\mu(C)$ of a cluster $C \subseteq S$ in the original graph $G$ can be calculated directly in the hypergraph model~$H_\mu$ using the ratio of its cut-net~$cut(C)$ to its weighted volume $\mathfrc{d}_{\mathfrc{w}}(C)$, assuming that the motif enumeration step is exact and $d_\mu(S) \leq d_\mu(\overline{S})$.
	\label{theo:conductance_equivalence}
\end{theorem}

\vspace*{-0.3cm}
\subsection{Flow Model.}
\label{subsec:Flow Model}

\begin{figure*}[t]
	\centering
	\includegraphics[width=.7\linewidth]{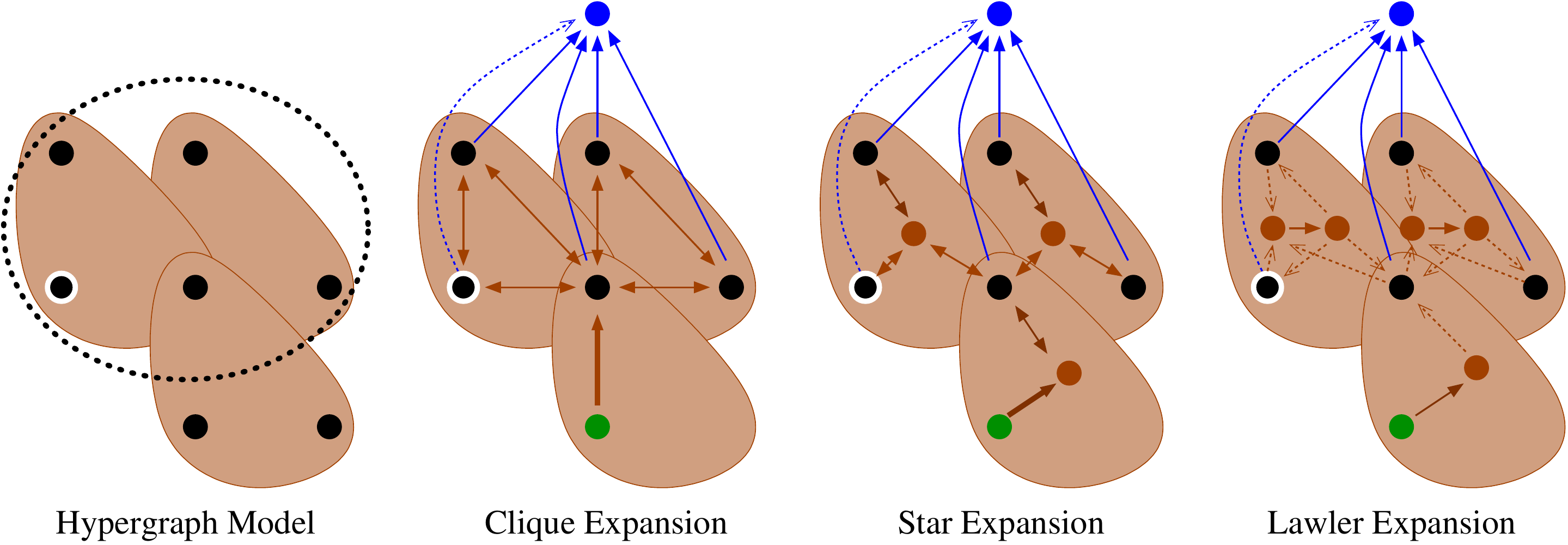}
	\caption{Flow model~$G_f$ given a hypergraph model~$H_\mu$ and an initial cluster~$C_0$.
		Nodes and nets of~$H_\mu$ are respectively represented by black circles and brown areas around them. 
		The seed node~$u$ is circled in white and the initial cluster~$C_0$ is surrounded by a dotted ellipse. 
		Auxiliary artificial nodes and edges used in each net-expansion are respectively represented by brown circles and arrows. 
		Bidirectional arrows represent pairs of edges in both directions.
		The seed node~$s$, the sink node~$t$, and the in-edges of~$t$ are respectively represented by a green circle, a blue circle, and blue arrows.
		Solid and dashed arrows respectively represent edges with finite and \hbox{infinite weight}.}
	\vspace*{-0.2cm}
	\label{fig:flow_model}
\end{figure*}

In this section, we describe the process of constructing our \texttt{MQI}-based flow model~$G_f$ using the hypergraph model~$H_\mu$ and an initial cluster~$C_0~\subseteq~S$ which contains the seed node~$u$. 
There are three possible implementations of~$G_f$ based on the three already explained techniques to represent hypergraphs using graphs, namely \emph{clique expansion}, \emph{star expansion}, and \emph{Lawler expansion} (see Figure~\ref{fig:hyperedge_expansion}).
We show a bijective correspondence between certain s-t cuts in $G_f$ and clusters~$C \subseteq C_0$~in~$G$ that include the seed node~$u$ and have motif conductance less~than~that~of~$C_0$.

We start by converting our hypergraph model~$H_\mu$ in a directed graph using the chosen net expansion technique.
Second, we find a corresponding cluster~$C_0^\prime$ for~$C_0$ in the created graph.
For the \emph{clique expansion}, $C_0^\prime~=~C_0$ since this transformation does not create artificial nodes.
For the \emph{star expansion},~$C_0^\prime$ consists of~$C_0$ and also the auxiliary artificial nodes connected to at least one node in~$C_0$.
For the \emph{Lawler expansion},~$C_0^\prime$ consists of~$C_0$, the auxiliary artificial nodes~$w_1$ having in-edges only from nodes in~$C_0$, and the auxiliary artificial nodes~$w_2$ having out-edges to at least one node in~$C_0$.
Third, we contract~$\overline{C_0^\prime}$ to a single \emph{source}~node~$s$ and then remove all its in-edges.
Fourth, we multiply the weight of all the remaining edges by $\mathfrc{d}_{\mathfrc{w}}(C_0)$, i.e., the weighted~volume of~$C_0$~in~$H_\mu$.
Fifth, we introduce a sink node~$t$ and include in-edges to it from each of the nodes~$v \in C_0\setminus\{u\}$, such that the weight of~$(v,t)$ is set to~$cut(C_{0})\mathfrc{d}_{\mathfrc{w}}(v)$, i.e., the cut-net of~$C_0$~in~$H_\mu$ multiplied by the weighted degree of~$v$~in~$H_\mu$.
Finally, we include an edge $(u,t)$ from the seed node to the sink and set its weight to infinity.
Our flow network model~$G_f$ is concluded by setting edge capacities to match edge weights.
Figure~\ref{fig:flow_model} shows the three possible configurations of our flow model~$G_f$ for a given hypergraph model~$H_f$ and an initial~cluster~$C_0$.

We now analyze the theoretical guarantees provided by the defined flow model~$G_f$.
Theorem~\ref{theo:flow_model_guarantee} shows that there is a set~${C\subset{C_0}}$ in~$G$ including the seed node~$u$ with motif conductance smaller than~that~of~$C_0$ if, and only if, the value of the maximum flow on $G_f$ is less than~$cut(C_{0})\mathfrc{d}_{\mathfrc{w}}(C_0)$, which is the weight of the trivial cut $(\{s\},V(G_f)\setminus\{s\})$.
In the proof, we show how such improved cluster~$C$ can be directly obtained from a maximum flow on~$G_f$.
For an even stronger claim, see Theorem~\ref{theo:flow_model_guarantee_gen} in the appendix. %
Assumptions (a)~and~(b) in Theorem~\ref{theo:flow_model_guarantee} are the same used in Theorem~\ref{theo:conductance_equivalence}, which were previously shown to be reasonable in practice.
Note that the claim is only valid for motifs with three nodes for clique and star expansion models, while it is valid in general for the Lawler expansion model.

\renewcommand{\labelenumi}{(\alph{enumi})}

\vspace*{-0.1cm}
\begin{theorem}
	There is a set~${C\subset{C_0}}$ in~$G$ including the seed node~$u$ with motif conductance smaller than that of~$C_0$ if, and only if, the maximum flow on~$G_f$ is less than $cut(C_{0})\mathfrc{d}_{\mathfrc{w}}(C_0)$ under the following assumptions:
	\begin{enumerate}
		\item the motif enumeration phase is exact;
		\item $d_\mu(S) \leq d_\mu(\overline{S})$~in~$G$;
		\item in case~$G_f$ is based on clique expansion or star expansion, the motif~$\mu$ has three nodes;
	\end{enumerate}
	\label{theo:flow_model_guarantee}
\end{theorem}

\begin{proof}
	We start with the backward direction, i.e., if the maximum flow on~$G_f$ is less than $cut(C_{0})\mathfrc{d}_{\mathfrc{w}}(C_0)$, then there is a sub-set of~${C_0}$ in~$G$ including the seed node~$u$ with motif conductance smaller than that of~$C_0$.
	According to the Max-Flow Min-Cut Theorem~\cite{ford_fulkerson_1956}, the weight of the maximum s-t flow on a network equals the weight of its minimum s-t cut, hence it follows that there is an s-t cut~$(B_1,B_2)$~of~$G_f$ with weight smaller than $cut(C_{0})\mathfrc{d}_{\mathfrc{w}}(C_0)$.
	Without loss of generality, let~$s \in B_1$, $t \in B_2$, $C = C_0 \cap B_2$, and hence, by definition,~${C_0\setminus{C}=C_0\cap{B_1}}$.
	Necessarily~${u\in{C}}$, otherwise the edge $(u,t)$, which has infinite weight, would be cut.
	There are two kinds of edges from $B_1$~to~$B_2$.
	First, there are the edges~$(x,t)$, with~$x \in C_0 \setminus C$.
	By design, the total weight of these edges is given by~$cut(C_{0})\mathfrc{d}_{\mathfrc{w}}(C_0\setminus{C})$.
	The second kind consists of edges from~$B_1$~to~${B_2\setminus\{t\}}$.
	These edges vary based on the net-expansion technique used, but their total weight is~$cut(C)\mathfrc{d}_{\mathfrc{w}}{({C}_0)}$ by design under \hbox{the~given~assumptions}. 
	
	Now we show that the total weight of the edges from~$B_1$~to~${B_2\setminus\{t\}}$ is~$cut(C)\mathfrc{d}_{\mathfrc{w}}{({C}_0)}$ for the three net-expansion techniques.
	In the clique expansion under assumption~(c), each cut net~$e$ of~$C$ in~$H_\mu$ corresponds directly to two cut in-edges of~$B_2$ in~$G_f$.
	The weight of each of these edges is, by design, set to $\mathfrc{w}(e)\mathfrc{d}_{\mathfrc{w}}{({C}_0)} / 2$, so they add up to the specified total weight.
	In the star expansion under assumption~(c), for each cut net~$e$ of~$C$ in~$H_\mu$ there is a single cut in-edge of~$B_2$ in~$G_f$ that connects an auxiliary artificial node and a node from $C_0$. 
	The weight of this cut edge is set by design to $\mathfrc{w}(e)\mathfrc{d}_{\mathfrc{w}}{({C}_0)}$, so the total weight of these edges is as claimed.
	In the Lawler expansion, for each cut net~$e$ of~$C$ in~$H_\mu$ there is exactly one cut in-edge of~$B_2$ in~$G_f$ that connects two auxiliary artificial nodes, namely~$w_1 \in B_1$~and~$w_2 \in B_2$. 
	If this were not the case, there would be an edge from~$B_1$~to~$B_2$ with infinite weight. 
	The weight of this single cut edge is set by design to $\mathfrc{w}(e)\mathfrc{d}_{\mathfrc{w}}{({C}_0)}$, so the sum of the weights of the cut edges is as stated.
	By adding up the weights of the two kinds of edges that cross the cut~$(B_1,B_2)$ and verifying that their total weight is less than $cut(C_{0})\mathfrc{d}_{\mathfrc{w}}(C_0)$, we derive Equation~(\ref{eq:forward_direction}), which can further be simplified to Equation~(\ref{eq:forward_direction_final}).
	We conclude the proof of the backward direction by applying Theorem~\ref{theo:conductance_equivalence}, Equation~(\ref{eq:forward_direction_final}), and the assumptions~(a)~and~(b).
	
	\vspace*{-0.2cm}
	\begin{align}
	\label{eq:forward_direction}
	cut(C_{0})\mathfrc{d}_{\mathfrc{w}}(C_0\setminus{C})~+
	&
	 \\ 
	cut(C)\mathfrc{d}_{\mathfrc{w}}{({C}_0)}
	& < cut(C_{0})\mathfrc{d}_{\mathfrc{w}}(C_0) \nonumber
	\end{align}
	
	\vspace*{-0.2cm}
	\begin{equation}
	\frac{cut(C)}{\mathfrc{d}_{\mathfrc{w}}(C)}
	<
	\frac{cut(C_0)}{\mathfrc{d}_{\mathfrc{w}}(C_0)}
	\label{eq:forward_direction_final}
	\end{equation}
	
	Now we prove the forward direction, i.e., given a set~$C~\subset~C_0$ including the seed node~$u$ in~$G$ with motif conductance smaller than~that~of~$C_0$, then the maximum flow on~$G_f$ is less than $cut(C_{0})\mathfrc{d}_{\mathfrc{w}}(C_0)$.
	With the assumptions~(a)~and~(b) in place, Equation~(\ref{eq:forward_direction_final}) holds true for the selected set~$C$.
	Since Equation~(\ref{eq:forward_direction_final}) can be rewritten as Equation~(\ref{eq:forward_direction}).
	To complete the proof, we will show that there exists an s-t~cut $(B_1,B_2)$ of $G_f$ such that $s \in B_1$, $t \in B_2$, $B_2 \cap C_0 = C$, and the total weight of the in-edges of $B_2$ is \hbox{$cut(C_{0})\mathfrc{d}{\mathfrc{w}}(C_0\setminus{C})+cut(C)\mathfrc{d}{\mathfrc{w}}{({C}_0)}$}.
	Using the Max-Flow Min-Cut Theorem~\cite{ford_fulkerson_1956} again, it follows that if there is an s-t cut with a weight of less than $cut(C_{0})\mathfrc{d}{\mathfrc{w}}(C_0)$, the maximum flow value on~$G_f$ must also be less than $cut(C_{0})\mathfrc{d}_{\mathfrc{w}}(C_0)$.
	Since $B_2 \cap C_0 = C$, it follows that $B_1 \cap C_0 = C_0 \setminus C$, hence there are $|C_0 \setminus C|$ cut edges of the form~$(x,t)$.
	By design, the total weight of these edges is $cut(C_{0})\mathfrc{d}{\mathfrc{w}}(C_0\setminus{C})$.
	
	Now we show that the weights of the remaining cut edges, i.e., edges from~$B_1$~to~${B_2\setminus\{t\}}$ add up to~$cut(C)\mathfrc{d}_{\mathfrc{w}}{({C}_0)}$ for each net-expansion technique. 
	In the clique expansion, we forcibly have $B_2=C \cup \{t\}$.
	Under assumption~(c), each cut net~$e$ of~$C$ in~$H_\mu$ corresponds to two cut in-edges of~$B_2$ in~$G_f$, both with weight set to $\mathfrc{w}(e)\mathfrc{d}_{\mathfrc{w}}{({C}_0)} / 2$ by design, hence they add up to the specified total weight.
	In the star expansion, under assumption~(c), we make $B_2 = C \cup A \cup \{t\}$, where~$A$ is the set of artificial nodes~$a$ with~$|N(a) \cap C| \geq |N(a)|/2$.
	For each cut net~$e$ of~$C$ in~$H_\mu$ there is a single cut in-edge of~$B_2$ in~$G_f$ that connects an auxiliary artificial node and a node from $C_0$. 
	The weight of this cut edge is set by design to $\mathfrc{w}(e)\mathfrc{d}_{\mathfrc{w}}{({C}_0)}$, so the total weight of these cut edges is as stated.
	In the Lawler expansion, we make~$B_2 = C \cup A \cup \{t\}$, where~$A$ consists of the auxiliary artificial nodes~$w_1$ having in-edges only from nodes in~$C$, and the auxiliary artificial nodes~$w_2$ having out-edges to at least one node in~$C$.
	Hence, for each cut net~$e$ of~$C$ in~$H_\mu$ there is exactly one cut in-edge of~$B_2$ in~$G_f$ that connects two auxiliary artificial nodes, namely~$w_1 \in B_1$~and~$w_2 \in B_2$. 
	The weight of this single cut edge is set by design to $\mathfrc{w}(e)\mathfrc{d}_{\mathfrc{w}}{({C}_0)}$, so~the~sum of the weights of the cut edges is as expected.
\end{proof}

\vspace*{+1.cm}
\texttt{SOCIAL} utilizes a push-relabel approach to iteratively search for a maximum s-t flow in the model~$G_f$.
If the found maximum flow is strictly smaller than $cut(C_{0})\mathfrc{d}_{\mathfrc{w}}(C_0)$, then we can directly find a minimum cut with the same weight as it and, consequently, a cluster~$C \subset C_0$ containing the seed node~$u$ that has a strictly smaller motif conductance value~$\phi_\mu(C)$ than that of~$C_0$~in~$G$. 
If such a cut is found, the algorithm repeats the process recursively setting the identified sub-cluster~$C$ as the new initial cluster, i.e., it constructs a new flow model based on~$H_\mu$ and the initial cluster and uses the push-relabel algorithm to continue searching for sub-clusters with even lower motif conductance values. 
If, on the other hand, the maximum flow is not strictly smaller than $cut(C_{0})\mathfrc{d}_{\mathfrc{w}}(C_0)$, it means that the current cluster~$C_0$ is optimal among all of its sub-clusters containing the seed node~$u$, and the \hbox{algorithm terminates for the given ball~$S$}.

\vfill

\section{Experimental Evaluation}
\label{sec:Experimental Evaluation}

\subsubsection*{Methodology.} 
We implemented \texttt{SOCIAL} in C++.
We compiled our program using gcc 11.2 with full optimization turned on (-O3 flag).
All our experiments are based on the triangle motif, i.e., the undirected clique of size three.
Since this motif has three nodes, Theorem~\ref{theo:flow_model_guarantee} is valid for all net expansion techniques. 
Therefore, we focus our experiments on the clique expansion technique, which is more efficient than the other techniques because it does not utilize any auxiliary artificial nodes and uses the minimum amount of auxiliary artificial edges.
We use the following parameters for \texttt{SOCIAL}: $\alpha=3$, $\ell \in \{1,2,3\}$.
We ensure the integrity of our results by using the same motif-conductance evaluator function for all tested algorithms.
In our experiments, we have used a machine with a sixty-four-core AMD EPYC 7702P processor running at $2.0$ GHz, $1$ TB of main memory, $32$ MB of L2-Cache, and $256$ MB of L3-Cache. 
We measure running time, motif-conductance, and/or size of the computed cluster.
For each graph, we pick $50$ random seed nodes and use all of them as input for each algorithm.
When averaging running time or cluster size over multiple instances, we use the geometric mean in order to give every instance the same influence on the \textit{final score}. 
When averaging motif conductance over multiple instances, the final score is computed via arithmetic mean.
This is a necessary averaging strategy since motif conductance can be zero, which makes the geometric mean infeasible to compute.
We also use \emph{performance profiles} which relate the running time (resp. motif conductance) of a group of algorithms to the fastest (resp. best) one on a per-instance basis.
Their x-axis shows a factor~$\tau$ while their y-axis shows the percentage of instances for which algorithm $A$ has up to~$\tau$ times the running time (resp. motif conductance) of the fastest (resp. best) algorithm.

\vspace*{-0.2cm}
\subsubsection*{Instances.}
The graphs used in our experiments are the same ones used by \hbox{\citet{yin2017local}} and \hbox{\citet{LocMotifClusHyperGraphPartition}} and the seed nodes used in our experiments are the same~ones~used~in~\cite{LocMotifClusHyperGraphPartition}.
Specifically, we use real graphs from the SNAP Network Dataset Collection \cite{snapnets}.
Prior to our experiments, we removed parallel edges, self-loops, and directions of edges and assigning unitary weight to all nodes and edges.
Basic properties of the graphs under consideration \hbox{can be found in Table~\ref{tab:graphs}}.

\begin{table}[t]
	\centering
	\setlength{\tabcolsep}{5.5pt}
	\footnotesize
	\begin{tabular}{ l  r  r  r  }
        \toprule
		Graph & $n$& $m$ & \# Triangles\\
	\midrule	
		com-amazon & \numprint{334863}  & \numprint{925872} & \numprint{667129} \\
		
		com-dblp & \numprint{317080}  & \numprint{1049866} & \numprint{2224385} \\
		
		com-youtube & \numprint{1134890}  & \numprint{2987624} & \numprint{3056386} \\
		
		com-livejournal & \numprint{3997962}  & \numprint{34681189} & \numprint{177820130} \\
		
		com-orkut & \numprint{3072441}  & \numprint{117185083} & \numprint{627584181} \\
		
		com-friendster & \numprint{65608366} & \numprint{1806067135} & \numprint{4173724142} \\		
                \bottomrule
		
	\end{tabular}
	\caption{Graphs for experiments.}
	\label{tab:graphs}
\end{table}

\begin{table*}[t]
	\centering
	\setlength{\tabcolsep}{8pt}
	\footnotesize
	\begin{tabular}{l@{\hskip 30pt}rrr@{\hskip 30pt}rrr@{\hskip 30pt}rrr}
	\toprule	
        \multirow{2}{*}{Graph} & \multicolumn{3}{c}{\hspace*{-1cm}\texttt{SOCIAL}}  & \multicolumn{3}{c}{\hspace*{-1cm}\texttt{LMCHGP}}    & \multicolumn{3}{c}{\hspace*{-1cm}\texttt{MAPPR}} \\
		& $\phi_\mu$     & $|C|$   & t(s)   & $\phi_\mu$     & $|C|$   & t(s)   & $\phi_\mu$  & $|C|$ & t(s)    \\ 
                \midrule

		com-amazon             & \textbf{0.031} & 76    & $<$0.01  & 0.037 & 64    & 0.22  & 0.153      & 58  & 2.68   \\
		com-dblp               & \textbf{0.090} & 58    & 0.02  & 0.115 & 56    & 0.38  & 0.289      & 35  & 3.04   \\
		com-youtube            & \textbf{0.125} & 1832  & 4.52  & 0.172 & 1443  & 7.93  & 0.910      & 2   & 10.44  \\
		com-livejournal        & \textbf{0.158} & 494   & 3.33  & 0.244 & 387   & 8.17  & 0.507      & 61  & 173.80 \\
		com-orkut              & 0.273 & 1041 & 256.21 & \textbf{0.150} & 13168 & 496.94 & 0.407      & 511 & 923.26 \\
		com-friendster         & 0.388      &  2060     & 1194.50       & \textbf{0.368}      &  10610     &   1339.99     &   0.741         &   121  &   16565.99      \\ 
                \midrule
		Overall                & \textbf{0.178} & 453   & 2.33  & 0.181 & 823   & 12.67  & 0.500      & 50  & 79.34  \\ 
                \bottomrule
	\end{tabular}
	\caption{Average comparison against state-of-the-art.}
	\label{tab:resultsoverall}
	\vspace*{-0.6cm}
\end{table*}

\vspace*{-0.2cm}
\subsubsection*{Competitors.}
We experimentally compare our \texttt{SOCIAL} against the state-of-the-art competitors, namely \texttt{MAPPR}~\cite{yin2017local} and the algorithm proposed by \hbox{\citet{LocMotifClusHyperGraphPartition}}.
For conciseness, we refer to the latter one from now on as~\mbox{\texttt{LMCHGP}},~an acronym for \emph{local motif clustering via (hyper)graph partitioning}.
We also ran preliminary experiments with \texttt{HOSPLOC}~\cite{zhou2021high}.
However, the algorithm very slow even for small graphs and not scalable as their algorithm works using an adjacency matrix and hence needs $\Omega(n^2)$ space and time. Moverover,  experiments done in their paper are on graphs that are  multiple orders of magnitude smaller than the graphs used in our evaluation.  Hence, we are not able to run the algorithm on the scale of the instances used in this work.  %
We were not able to explicitly compare against~\texttt{LCD-Motif}~\cite{zhang2019local} since their code is not available (neither public, nor privately\footnote{Personal communication with the authors}) and the data presented in the respective paper does not warrant explicit comparisons (e.g. seed nodes are typically not presented in the papers).
However, we try to make implicit comparisons \hbox{in Section~\ref{subsubsec:addcomp}}. %

We compare our results against the globally best cluster computed for each seed node by \texttt{MAPPR} using its standard parameters ($\alpha=0.98$, $\epsilon=10^{-4}$) and by \mbox{\texttt{LMCHGP}} using the configuration with best overall results in~\cite{LocMotifClusHyperGraphPartition} (graph model, label propagation, $\alpha=3$, $\ell \in \{1,2,3\}$, and $\beta=80$).
Unless mentioned otherwise, experiments presented here involve all \hbox{graphs~from~Table~\ref{tab:graphs}}.

\vspace*{-0.2cm}
\subsection{Results.}
\label{subsec:Results}

In this section we present and discuss our results.
In the performance profile plots shown in Figures~\ref{fig:SIMPLEstateoftheart_graph_res_pp}~and~\ref{fig:SIMPLEstateoftheart_graph_tim_pp}, we compare \mbox{\texttt{LMCHGP}}~\cite{LocMotifClusHyperGraphPartition} and \texttt{MAPPR}~\cite{yin2017local} against \texttt{SOCIAL}.
In Table~\ref{tab:resultsoverall}, we show average results for each graph in our Test Set \hbox{as well as average~results~overall}.

As shown in Figure~\ref{fig:SIMPLEstateoftheart_graph_res_pp}, \texttt{SOCIAL} obtains the best or equal motif conductance value for $62\%$ of the instances, while \mbox{\texttt{LMCHGP}} and \mbox{\texttt{MAPPR}} respectively obtain the best or equal motif conductance for $49\%$~and~$19\%$ of the instances.
This result can be explained with two observations.
First, \texttt{SOCIAL} explores the solution space considerably better than \texttt{MAPPR}, since we build our model multiple times, while \texttt{MAPPR} simply uses the  \texttt{APPR} algorithm.
Second, \texttt{SOCIAL} is based on a flow approach which directly optimizes for motif conductance, whereas \mbox{\texttt{LMCHGP}} is based on a (hyper)graph partitioning algorithm which is repeated multiple times to compensate for its design to minimize the number of cut motifs rather than motif conductance.
In Table~\ref{tab:resultsoverall}, \texttt{SOCIAL} outperforms \mbox{\texttt{LMCHGP}} for 4 of the 6 graph and overall, and outperforms \mbox{\texttt{MAPPR}} for all graphs and overall.
Overall, \texttt{SOCIAL} computes clusters with motif conductance $0.178$ while \mbox{\texttt{LMCHGP}} and \mbox{\texttt{MAPPR}} compute clusters with motif conductance $0.181$~and~$0.500$,~respectively.

\begin{figure}[t]{}
	\centering

	\vspace*{-.6cm}
	\includegraphics[width=0.9\linewidth]{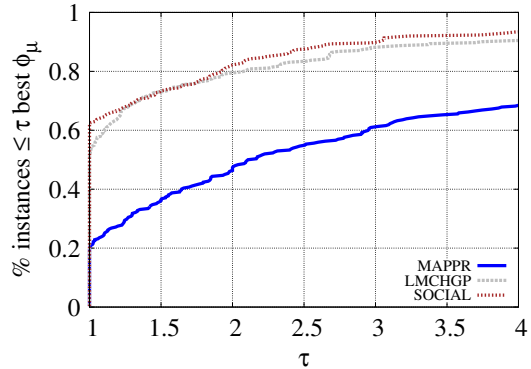}
	\vspace*{-0.75cm}
	\caption{Motif conductance performance profile.}
                \vspace*{-.75cm}
	\label{fig:SIMPLEstateoftheart_graph_res_pp}
\end{figure}

As exhibited in Figure~\ref{fig:SIMPLEstateoftheart_graph_tim_pp}, \texttt{SOCIAL} is the fastest one for $87\%$ of the instances, while \mbox{\texttt{LMCHGP}} and \mbox{\texttt{MAPPR}} are the fastest ones for $12\%$~and~$1\%$ of the instances, respectively.
Furthermore, the running time of \texttt{SOCIAL} is within a factor $1.18$ of the running times of the fastest competitors for all instances.
\texttt{SOCIAL} is respectively up to~\numprint{237}~and~\numprint{144063}~times faster than \mbox{\texttt{LMCHGP}} and \mbox{\texttt{MAPPR}} while being a factor~$5.4$~and~$34.1$ faster than them on average.
The reason for \texttt{MAPPR} being considerably slower than the other algorithms is that it must enumerate motifs across the entire graph, while \texttt{SOCIAL} and \texttt{LMCHGP} only require enumeration of motifs in a ball around the seed node.
The reduced but still substantial difference in running time between \texttt{SOCIAL} and \texttt{LMCHGP} is a result of \texttt{LMCHGP}'s repeated partitioning of each ball around the seed node, while \texttt{SOCIAL} employs a flow model to greedily improve the motif conductance metric until a local optimum cluster is obtained.
In Table~\ref{tab:resultsoverall}, \texttt{SOCIAL} outperforms \mbox{\texttt{LMCHGP}} and \mbox{\texttt{MAPPR}} on average in terms of running time \hbox{for every single graph~and~overall}.

For a more intuitive analysis of the quality of our results, Figure~\ref{fig:SIMPLEstateoftheart_graph_res_double} plots motif conductance versus cluster size for all communities computed by the three algorithms.
Observe that the communities found by \texttt{SOCIAL} are densely localized in the lower left area of the chart, which is the region with smaller motif conductance and smaller cluster size. 
On the other hand, communities computed by \texttt{MAPPR} are often in the upper area of the chart and communities computed by \mbox{\texttt{LMCHGP}} are often in the right area~of~the~chart.

\begin{figure}[t]{}
	\centering
	\includegraphics[width=0.9\linewidth]{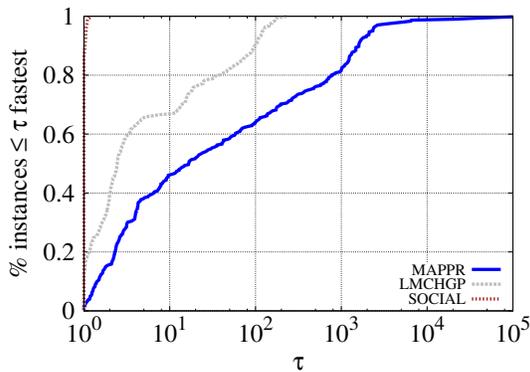}
	\vspace*{-0.5cm}
	\caption{Running time performance profile.}
	\vspace*{-0.25cm}
	\label{fig:SIMPLEstateoftheart_graph_tim_pp}
\end{figure}

\begin{figure}[t]{}
	\centering
	\includegraphics[width=0.9\linewidth]{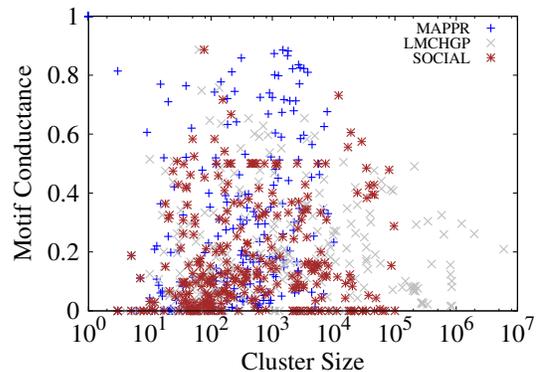}
	\vspace*{-0.5cm}
	\caption{Motif conductance vs cluster size.}
	\vspace*{-0.25cm}
	\label{fig:SIMPLEstateoftheart_graph_res_double}
\end{figure}

\subsubsection*{Additional Comparisons.}
\label{subsubsec:addcomp}
As mentioned above, we were not able to compare against \texttt{LCD-Motif}~\cite{zhang2019local} explicitly since their code is not available (neither publicly, nor privately) and the data presented in the respective papers does not warrant explicit comparisions (e.g.,~seed nodes are typically not presented in papers, and in this case instances are directed rather than undirected).
Here, we make an attempt at implicit comparisons.  Zhang~\etal~\cite{zhang2019local} (Table 4 therein) compare motif conductance against \texttt{MAPPR} on three directed instances (cit-HepPh, Slashdot, StanfordWeb) and report an geometric mean improvement of 54\% in motif conductance for the triangle motif. 
As \texttt{SOCIAL} works for undirected instances, we have build undirected version of those graphs and run \texttt{SOCIAL} as well as \texttt{MAPPR} for the triangle motif. The geometric mean improvement we obtain over \texttt{MAPPR} is 223\% %
which is significantly larger than the improvement of Zhang~\etal over \texttt{MAPPR}. Also note that in our experiments from Table~\ref{tab:resultsoverall}, the geometric mean improvement (using the average motif conductance values) of \texttt{SOCIAL} over \texttt{MAPPR} in \hbox{motif conductance is 219\%}.%

\section{Conclusion}
\label{sec:Conclusion}

In this work, we propose \texttt{SOCIAL}, a fast flow-based algorithm to solve the local motif clustering problem in graphs. 
Given a seed node, our \texttt{SOCIAL} selects a ball of nodes around it, which is taken as an initial cluster and used to build an exact hypergraph model where nets represent motifs.
Using this model and the initial cluster, we create a flow model in which the value of the maximum s-t flow is directly related to the presence of sub-sets of the initial cluster that contain the seed node and have lower motif conductance than the initial cluster as a whole.
Utilizing a push-relabel algorithm, \texttt{SOCIAL} either identifies a sub-cluster containing the seed node with improved motif conductance and repeats the process recursively by considering it as the initial cluster, or demonstrates that the current initial cluster is the best among all its sub-clusters \hbox{that include the seed node}.

In our experiments with the triangle motif, we found that \texttt{SOCIAL} produces communities with an average motif conductance better than the state-of-the-art, while running up to orders of magnitude faster~on~average.
Given the good results of our algorithm, we plan to release it as open source soon.
For future work, we intend to conduct experiments with larger motifs and use the Lawler-expansion version of our flow graph, since it is the only one whose quality guarantee holds true for larger motifs. Laslty, we intend to add parallelization to improve the speed on large instances further.

\vfill
\pagebreak

\bibliographystyle{ACM-Reference-Format}
\bibliography{phdthesiscs}

\appendix
\section{A More General Theorem.}
In this section, we provide Theorem~\ref{theo:flow_model_guarantee_gen}, which establishes a one-to-one correspondence between sets~${C\subset{C_0}}$ in~$G$ including the seed node~$u$ with motif conductance smaller than~that~of~$C_0$ and cuts of $G_f$ with weight less than~$cut(C_{0})\mathfrc{d}_{\mathfrc{w}}(C_0)$, which is the weight of the trivial cut $(\{s\},V(G_f)\setminus\{s\})$.
This correspondence depends upon the same assumptions made in Theorem~\ref{theo:conductance_equivalence}, which were previously shown to be reasonable in practice.
Nevertheless, note that the claim is only valid for motifs with three nodes for clique expansion and star expansion-based models, while it is valid in general for the Lawler expansion-based model.
Although the model~$G_f$ based on Lawler expansion requires the largest number of auxiliary artificial nodes, its exactness for motifs of any size makes it a powerful and widely applicable approach for local motif clustering.
It is worth mentioning that assumption~(d) in Theorem~\ref{theo:flow_model_guarantee_gen} does not restrict the validity of the star expansion model. 
This is because any cut of the model with finite weight either satisfies this assumption or can be slightly altered to comply with it without moving any nodes from~$C_0$ to a different block and while strictly decreasing the total weight of the cut.

\begin{theorem}
	The collection of s-t cuts~$(B_1,B_2)$ of~$G_f$ where the total weight of the edges from~$B_1$~to~$B_2$ is smaller than $cut(C_{0})\mathfrc{d}_{\mathfrc{w}}(C_0)$ and the collection of sets~${C\subset{C_0}}$ in~$G$ including the seed node~$u$ with motif conductance smaller than that of~$C_0$ are bijectively related under the following assumptions:
	\begin{enumerate}
		\item the motif enumeration phase is exact;
		\item $d_\mu(S) \leq d_\mu(\overline{S})$~in~$G$;
		\item in case~$G_f$ is based on clique expansion or star expansion, the motif~$\mu$ has three nodes;
		\item in case~$G_f$ is based on star expansion, artificial nodes are in blocks including most of its neighbors.
	\end{enumerate}
	\label{theo:flow_model_guarantee_gen}
\end{theorem}

\begin{proof}
	Without loss of generality, let~$s \in B_1$~and~$t \in B_2$.
	Let~$C = C_0 \cap B_2$ hence, by definition,~${C_0\setminus{C}=C_0\cap{B_1}}$.

	We start with the forward direction, i.e., given any s-t cut~$(B_1,B_2)$~of~$G_f$ with weight smaller than $cut(C_{0})\mathfrc{d}_{\mathfrc{w}}(C_0)$, we can obtain a unique set~$C~\subset~C_0$ in~$G$, including~$u$, with motif conductance smaller than~that~of~$C_0$.
	Let~$C = C_0 \cap B_2$ hence, by definition,~${C_0\setminus{C}=C_0\cap{B_1}}$.
	Necessarily~${u\in{C}}$, otherwise the edge $(u,t)$, which has infinite weight, would be cut.
	There are two kinds of edges from $B_1$~to~$B_2$.
	First, there are the edges~$(x,t)$, with~$x \in C_0 \setminus C$.
	By design, the total weight of these edges is given by~$cut(C_{0})\mathfrc{d}_{\mathfrc{w}}(C_0\setminus{C})$.
	The second kind consists of edges from~$B_1$~to~${B_2\setminus\{t\}}$.
	These edges vary based on the net-expansion technique used, but their total weight is~$cut(C)\mathfrc{d}_{\mathfrc{w}}{({C}_0)}$ by design under the given assumptions. 
	
	Now we show that the total weight of the edges from~$B_1$~to~${B_2\setminus\{t\}}$ is~$cut(C)\mathfrc{d}_{\mathfrc{w}}{({C}_0)}$ for the three net-expansion techniques.
	In the clique expansion under assumption~(c), each cut net~$e$ of~$C$ in~$H_\mu$ corresponds directly to two cut in-edges of~$B_2$ in~$G_f$.
	The weight of each of these edges is, by design, set to $\mathfrc{w}(e)\mathfrc{d}_{\mathfrc{w}}{({C}_0)} / 2$, so they add up to the specified total weight.
	In the star expansion under assumptions~(c)~and~(d), for each cut net~$e$ of~$C$ in~$H_\mu$ there is a single cut in-edge of~$B_2$ in~$G_f$ that connects an auxiliary artificial node and a node from $C_0$. 
	The weight of this cut edge is set by design to $\mathfrc{w}(e)\mathfrc{d}_{\mathfrc{w}}{({C}_0)}$, so the total weight of these edges is as claimed.
	In the Lawler expansion, for each cut net~$e$ of~$C$ in~$H_\mu$ there is exactly one cut in-edge of~$B_2$ in~$G_f$ that connects two auxiliary artificial nodes, namely~$w_1 \in B_1$~and~$w_2 \in B_2$. 
	If this were not the case, there would be an edge from~$B_1$~to~$B_2$ with infinite weight. 
	The weight of this single cut edge is set by design to $\mathfrc{w}(e)\mathfrc{d}_{\mathfrc{w}}{({C}_0)}$, so the sum of the weights of the cut edges is as stated.
	By adding up the weights of the two kinds of edges that cross the cut~$(B_1,B_2)$ and verifying that their total weight is less than $cut(C_{0})\mathfrc{d}_{\mathfrc{w}}(C_0)$, we derive Equation~(\ref{eq:forward_direction}), which can further be simplified to Equation~(\ref{eq:forward_direction_final}).
	We conclude the proof of the forward direction of the bijection by applying Theorem~\ref{theo:conductance_equivalence}, Equation~(\ref{eq:forward_direction_final}), and the assumptions~(a)~and~(b).
	
	Now we prove the backward direction of the bijection, i.e., given any set~$C~\subset~C_0$ including the seed node~$u$ in~$G$ with motif conductance smaller than~that~of~$C_0$, we can obtain a unique s-t cut~$(B_1,B_2)$~of~$G_f$ with weight smaller than $cut(C_{0})\mathfrc{d}_{\mathfrc{w}}(C_0)$.
	With the assumptions~(a)~and~(b) in place, Equation~(\ref{eq:forward_direction_final}) holds true for the selected set~$C$.
	Since Equation~(\ref{eq:forward_direction_final}) can be rewritten as Equation~(\ref{eq:forward_direction}), to complete the proof we must show that there exists a unique set $B_2$ such that $B_2 \cap C_0 = C$ and the total weight of the in-edges of $B_2$ is $cut(C_{0})\mathfrc{d}{\mathfrc{w}}(C_0\setminus{C}) + cut(C)\mathfrc{d}{\mathfrc{w}}{({C}_0)}$.
	Since $B_2 \cap C_0 = C$, we can deduce that $B_1 \cap C_0 = C_0 \setminus C$, hence there are $|C_0 \setminus C|$ cut edges of the form~$(x,t)$.
	By design, the total weight of these edges is $cut(C_{0})\mathfrc{d}{\mathfrc{w}}(C_0\setminus{C})$.
	
	Now we show that the weights of the remaining cut edges, i.e., edges from~$B_1$~to~${B_2\setminus\{t\}}$ add up to~$cut(C)\mathfrc{d}_{\mathfrc{w}}{({C}_0)}$ for each net-expansion technique. 
	In the clique expansion, we forcibly have $B_2=C \cup \{t\}$.
	Under assumption~(c), each cut net~$e$ of~$C$ in~$H_\mu$ corresponds to two cut in-edges of~$B_2$ in~$G_f$, both with weight set to $\mathfrc{w}(e)\mathfrc{d}_{\mathfrc{w}}{({C}_0)} / 2$ by design, hence they add up to the specified total weight.
	In the star expansion, under assumptions~(c)~and~(d), we forcibly have $B_2 = C \cup A \cup \{t\}$, where~$A$ is the set of artificial nodes~$a$ with~$|N(a) \cap C| \geq |N(a)|/2$.
	For each cut net~$e$ of~$C$ in~$H_\mu$ there is a single cut in-edge of~$B_2$ in~$G_f$ that connects an auxiliary artificial node and a node from $C_0$. 
	The weight of this cut edge is set by design to $\mathfrc{w}(e)\mathfrc{d}_{\mathfrc{w}}{({C}_0)}$, so the total weight of these cut edges is as stated.
	In the Lawler expansion, we forcibly have~$B_2 = C \cup A \cup \{t\}$, where~$A$ consists of the auxiliary artificial nodes~$w_1$ having in-edges only from nodes in~$C$, and the auxiliary artificial nodes~$w_2$ having out-edges to at least one node in~$C$.
	Hence, for each cut net~$e$ of~$C$ in~$H_\mu$ there is exactly one cut in-edge of~$B_2$ in~$G_f$ that connects two auxiliary artificial nodes, namely~$w_1 \in B_1$~and~$w_2 \in B_2$. 
	The weight of this single cut edge is set by design to $\mathfrc{w}(e)\mathfrc{d}_{\mathfrc{w}}{({C}_0)}$, so the sum of the weights of the cut edges is as expected.
\end{proof}

\end{document}